\documentclass[11pt,letterpaper]{article}

\usepackage[margin=1in]{geometry}
\usepackage[utf8]{inputenc}

\usepackage{cite}
\usepackage[bookmarks=true,pdfstartview=FitH,colorlinks,linkcolor=darkred,filecolor=darkred,citecolor=darkred,urlcolor=darkred]{hyperref}

\usepackage{multicol}
\usepackage{color,xcolor}
\usepackage{graphicx,color,eso-pic}
\usepackage{amsmath,amssymb,stmaryrd}
\usepackage{boxedminipage}
\usepackage{cite}
\usepackage{url}
\usepackage{graphicx}
\usepackage[bf]{caption}
\usepackage{subcaption}
\usepackage{color}
\usepackage{xspace}
\usepackage{multirow}
\usepackage{amsmath,amsthm,amstext,amssymb,amsfonts,latexsym}
\usepackage{wrapfig}
\usepackage{comment}
\usepackage{algorithmicx}
 \usepackage{algpseudocode}

\definecolor{darkred}{rgb}{0.5, 0, 0}
\definecolor{darkgreen}{rgb}{0, 0.5, 0}
\definecolor{darkblue}{rgb}{0,0,0.5}

\newtheorem{thm}{Theorem}[section]      

\newtheorem{lemma}[thm]{Lemma}

\newtheorem{corollary}[thm]{Corollary}

\newtheorem{definition}[thm]{Definition}

\newtheoremstyle{boxes}
{2pt}
{0pt}
{}
{}
{\bfseries}
{}
{\newline}
{\thmname{#1}\thmnumber{ #2}:  
\thmnote{#3}}

\theoremstyle{boxes}
\newtheorem{algorithm}[thm]{Algorithm}
\newtheorem{myalgo}[thm]{Algorithm}
\newtheorem{myfunc}[thm]{Functionality}
\newtheorem{myconst}[thm]{Construction}

\renewcommand{\sec}{{\lambda}}
\newcommand{\A}{{\cal A}}

\newcommand{\bit}{{\{0,1\}}}

\newcommand{\X}{\mathbf{X}}
\renewcommand{\sec}{{\lambda}}
\newcommand{\e}{{\epsilon}}

\newcommand{\bsize}{\ensuremath{b}}
\newcommand{\op}{\ensuremath{\mathsf{op}}\xspace}
\newcommand{\data}{{\ensuremath{\mathsf{data}}}\xspace}
\newcommand{\addr}{{\ensuremath{\mathsf{addr}}}\xspace}

\newcommand{\Sim}{\ensuremath{{\sf Sim}}\xspace}
\newcommand{\abs}[1]{\left|{#1}\right|}

\newcommand{\mem}{{\ensuremath{\mathsf{mem}}}\xspace}

\newcommand{\Read}{{\ensuremath{\mathsf{read}}}\xspace}
\newcommand{\Write}{{\ensuremath{\mathsf{write}}}\xspace}

\begin{document}

\title{\bf Bucket Oblivious Sort: \\An Extremely Simple Oblivious Sort\thanks{The paper was presented in the 3rd Symposium on Simplicity in Algorithms, SOSA@SODA 2020. This version is identical
to the SOSA'20 conference version modulo typo corrections.
}}

\author{Gilad Asharov \thanks{Bar-Ilan University. Part of the work was done while the author was a post-doctoral fellow at Cornell Tech supported a Junior Fellow award from the Simons Foundation, and while at J.P. Morgan AI Research.} \qquad
T-H. Hubert Chan\thanks{The University of Hong Kong. Partially supported the Hong Kong RGC under the grant 17200418.} \qquad
Kartik Nayak\thanks{Duke University. Part of the work was done while the author was at University of Maryland.} \\
Rafael Pass\thanks{Cornell Tech.} \qquad
Ling Ren\thanks{University of Illinois Urbana-Champaign. Part of the work was done while the author was at MIT.} \qquad
Elaine Shi\thanks{Cornell University.}}

\newcommand{\rl}[1]{{\footnotesize\color{orange}[Ling: #1]}}

\date{}

\maketitle


\begin{abstract}
We propose a conceptually simple oblivious sort and oblivious random permutation algorithms called bucket oblivious sort and bucket oblivious random permutation.
Bucket oblivious sort uses $6n\log n$ time (measured by the number of memory accesses) and $2Z$ client storage with an error probability exponentially small in $Z$. 
The above runtime is only $3\times$ slower than a non-oblivious merge sort baseline;
for $2^{30}$ elements, it is $5\times$ faster than bitonic sort,
the de facto oblivious sorting algorithm in practical implementations. 



\end{abstract}


\section{Introduction}

With the increased use of outsourced storage and computation, privacy of the outsourced data has been of paramount importance. 
A canonical setting is where a client with a small local
storage outsources its encrypted data to an untrusted server. 
In this setting, encryption alone is not sufficient to preserve privacy.
The access patterns to the data may reveal sensitive information. 

Two fundamental building blocks for oblivious storage and computation~\cite{goldreich1996software,goodrich2011privacy,oblivistore} are oblivious sorting and oblivious random permutation.
In these two problems, an array of $n$ elements is stored on an untrusted server, encrypted under a trusted client's secret key.
The client wishes to sort or permute the $n$ elements in a \emph{data-oblivious} fashion.
That is, the sequence of accesses it makes to the server should not reveal any information about the $n$ elements (e.g., their relative ranking).
The client has a small amount of local storage, the access pattern to which cannot be observed by the server.
This work presents simple and efficient algorithms to these two problems, named bucket oblivious sort and bucket oblivious random permutation.

\subsection{State of the Affairs}
For oblivious sort, it is well-known that one can leverage 
sorting networks such as AKS~\cite{aks} and Zig-zag sort~\cite{zigzag}
to obliviously sort $n$ elements in $O(n\log n)$ time. 
Unfortunately, these algorithms are complicated and incur enormous constants rendering them completely impractical. 
Thus, almost all known practical implementations~\cite{oblivistore,oblivm,graphsc}
instead employ the simple bitonic sort algorithm~\cite{bitonic}. 
While asymptotically worse, due to the small leading constants, bitonic sort performs much better in practice.

Oblivious random permutation (ORP) can be realized by assigning a sufficiently long random key to each element, and then obliviously sorting the elements by the keys.
To the best of our knowledge, this remains the most practical solution for ORP.
It then follows that while $O(n \log n)$ algorithms exist in theory, practical instantiations resort to the $O(n \log^2 n)$ bitonic sort.
There exist algorithms such as the Melbourne
shuffle~\cite{ohrimenko2014melbourne} that do not rely on
oblivious sort; but they require $O(\sqrt{n})$ client storage to permute $n$ elements.
Other approaches include the famous Thorp shuffle~\cite{thorp01} and random permutation networks~\cite{randpermnet}, but none of these solutions are competitive in performance either asymptotically or concretely.

\subsection{Our Results}
Let $Z$ be a statistical security parameter that controls the error probability. 
Our bucket oblivious sort runs in $6n\log n$ time ($4n\log n$ for bucket ORP) and has an error probability around $e^{-Z/6}$ when the client can store $2Z$ elements locally.
This is at most $3\times$ slower than the non-oblivious merge sort, and is at least $5\times$ faster than bitonic sort for $n=2^{30}$ (cf. Table~\ref{tab:compare}).
Therefore, we recommend bucket oblivious sort and bucket ORP as attractive alternatives to bitonic sort in practical implementations.

\begin{figure*}[h!]
\centering
\includegraphics[width=0.7\textwidth]{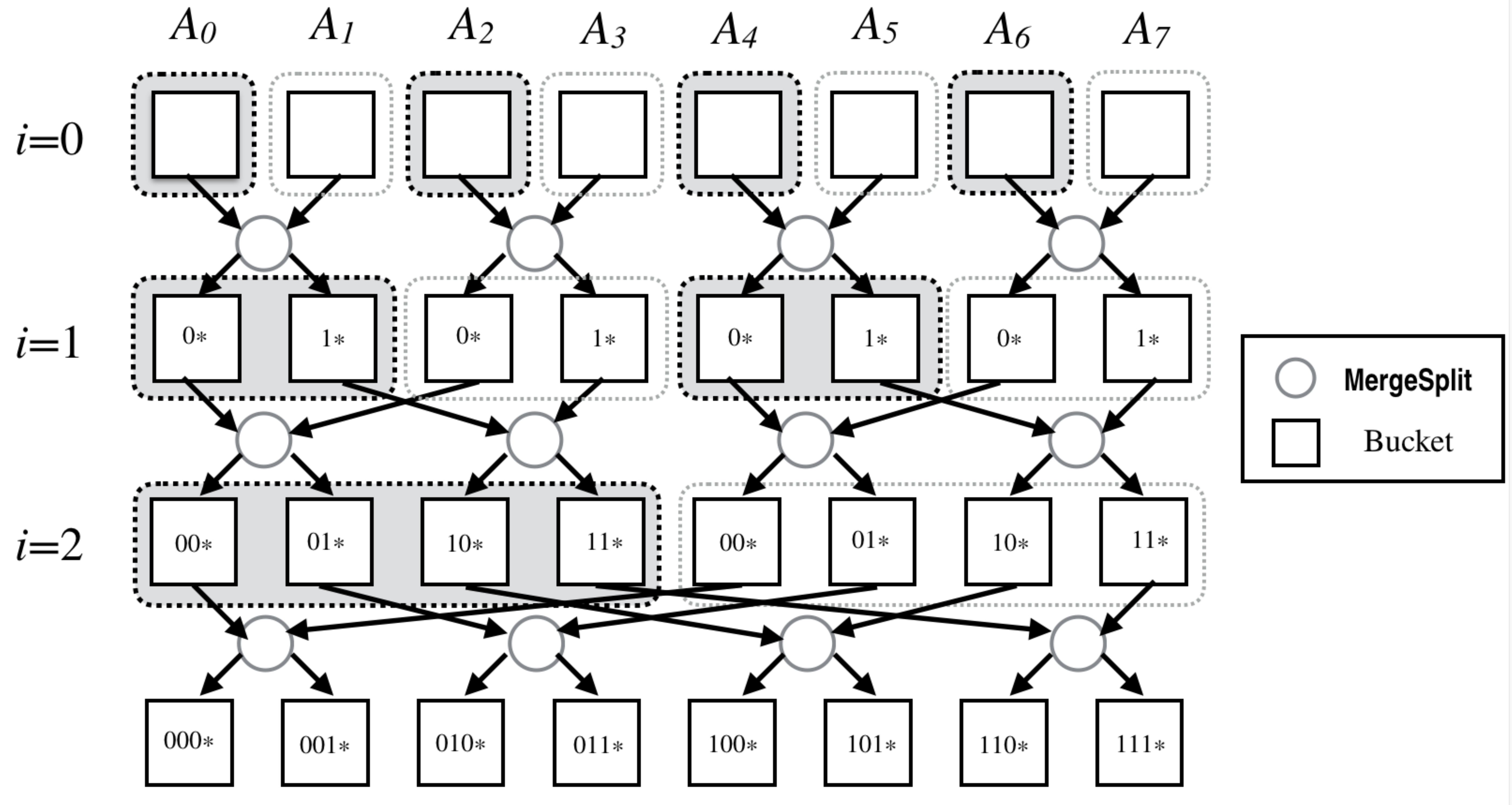}
\captionof{figure}{\textbf{Oblivious random bin assignment with 8 buckets.}
The \textsc{MergeSplit} procedure takes elements from two buckets at level $i$ and put them into two buckets at level $i+1$, according to the $(i+1)$-th most significant bit of the keys. 
At level $i$, every $2^{i}$ consecutive buckets are semi-sorted by the most significant $i$ bits of the keys.}
\label{fig:radix-sort}

\bigskip
\centering
\begin{tabular}{|c|c|c|c|c|}
    \hline
    Algorithm & Oblivious & Client storage & Runtime & Error probability \\
    \hline
    Merge sort & No & $O(1)$ & $2n \log n$ & 0 \\
    Bitonic sort & Yes & $O(1)$ & $n\log^2 n$ & 0 \\
    AKS sort~\cite{aks} & Yes & $O(1)$ & $5.4\times10^7 \times n\log n$ & 0 \\
    Zig-zag sort~\cite{zigzag} & Yes & $O(1)$ & $8\times10^4 \times n\log n$ & 0 \\
    Randomized Shellsort~\cite{RandShellsort} & Yes & $O(1)$ & $24n\log n$ & $\approx n^{-3}$ \\
    \hline
    \textbf{Bucket oblivious sort} & Yes & $2Z$ & $6 n\log n$ & $\approx e^{-Z/6}$\\
    \textbf{Bucket oblivious sort} & Yes & $O(1)$ & $\approx 2n\log n \log^2 Z$ & $\approx e^{-Z/6}$ \\
    \hline   
\end{tabular}
\captionof{table}{\textbf{Runtime of bucket oblivious sort and
classic non-oblivious and oblivious sort algorithms.} Bitonic
sort requires $\frac 14 n \log^2 n$ comparisons. The number of comparisons for AKS
sort and zig-zag sort are cited from \cite{zigzag}. Runtime
represents the number of memory accesses, which is four times the
number of comparisons.}
\label{tab:compare}

\end{figure*}

The core of our algorithms is to assign each element to a random bin and then route the elements through a butterfly network to their assigned random bins. 
This part is inspired by Bucket ORAM~\cite{fletcher2015bucket}. 
In more detail, we divide the $n$ elements into $B=2n/Z$ buckets of size $Z/2$ each and add $Z/2$ dummy elements to each bucket.
Now, imagine that these $B$ buckets form the inputs of a butterfly network --- for simplicity, assume $B$ is a power of two.
Each element is uniformly randomly assigned to one of the $B$ output buckets, represented by a key of $\log B$ bits.
The elements are then routed through the butterfly network to their respective destinations.
Assuming the client can store two buckets locally at a time, at level $i$, the client simply reads elements from two buckets that are distance $2^i$ away in level $i$ and writes them to two adjacent buckets in level $i+1$, using the $i$-th bit of each element's key to make the routing decision. 
We refer readers to Figure~\ref{fig:radix-sort} for a graphical illustration.

The above algorithm is clearly oblivious, as the order in which the client reads and writes the buckets is fixed and independent of the input array. If no bucket overflows, all elements reach their assigned destinations. By setting $Z$ appropriately, we can bound the overflow probability.

Our bucket oblivious sort and bucket ORP algorithms are derived from
the above oblivious random bin assignment building block. 

\paragraph{From oblivious random bin assignment to ORP and oblivious sort.}
To obtain a random permutation, we simply remove all dummy elements and randomly permute 
each bucket of the final layer.
Since the client can hold $Z$ elements, permuting each bucket can be done locally. 
We show that the algorithm is oblivious and gives a random permutation despite revealing the number of dummy elements in each destination bucket.
To get oblivious sort, we can first perform ORP on the input array then apply any \emph{non-oblivious, comparison-based} sorting algorithm (e.g., quick sort or
merge sort). We show that the composition of ORP and non-oblivious sort results in an oblivious sort. 

\paragraph{Dealing with small client storage.}
In Section~\ref{sec:O1client}, we extend our algorithms to support $O(1)$ client storage. 
We can rely on bitonic sort to realize the \textsc{MergeSplit} operation that operates on 4 buckets at a time,
which would result in $O(n\log n\cdot \log^2 Z)$ runtime. 

\paragraph{Locality.} Algorithmic performance when the data is stored on disk has been studied in the external disk model (e.g.,~\cite{RuemmlerW94,ArgeFGV97,Vitter01,Vitter06}) and references within). Recently, Asharov et al.~\cite{AsharovCNPRS19} extended this study to oblivious algorithms.  We discuss how our algorithms can be made locality-friendly in Section~\ref{sec:locality}. 

\paragraph{Subsequent work.}
The work by Ramachandran and Shi~\cite{Ramachandran} improved the algorithm in a cache-oblivious, cache-efficient manner in a binary fork-join model of computation.

\newcommand{\memsize}{{N}}
\newcommand{\blocksize}{{b}}
\newcommand{\Y}{{\bf Y}}

%
\section{Preliminaries}
\label{sec:defs}

\paragraph{Notations and conventions.}
Let $[n]$ denote the set $\{1,\ldots,n\}$. Throughout this paper, we will use
$n$ to denote the size of the instance and use $\lambda$ to denote the security parameter. 
For an ensemble of distributions $\{D_\sec\}$ (parametrized with $\sec$),
we denote by $x \leftarrow D_\sec$ a sampling of an instance from the distribution $D_\sec$. 
We say two ensembles of distributions $\{X_\sec\}$~and~$\{Y_\sec\}$ 
are $\e(\sec)$-statistically-indistinguishable, denoted $\{X_\sec\} \overset{\e(\sec)}{\equiv} \{Y_\sec\}$, 
if for any unbounded adversary $\A$, 
\[
\left|\Pr_{x\leftarrow X_\sec}\left[\A(1^\lambda, x)=1\right] - \Pr_{y\leftarrow Y_\sec}\left[\A(1^\lambda, y)=1\right] \right| \leq \e(\sec) \ .
\]

\paragraph{Random-access machines.}
A RAM is an interactive Turing machine that consists of a memory and a CPU.  The
memory is denoted as $\mem[\memsize,\blocksize]$, and is indexed by the logical
address space $[N] = \{1,2,\ldots,N\}$. We refer to each memory word also as a
\emph{block} and we use $\bsize$ to denote the bit-length of each block. The memory supports read/write
instructions $(\op,\addr, \data)$, where $\op \in \{\Read,\Write\}$, $\addr \in
[N]$ and $\data \in \bit^\bsize \cup \{\bot\}$.  If $\op = \Read$, then
$\data=\bot$ and the returned value is the content of the block located in
logical address $\addr$ in the memory. If $\op=\Write$, then the memory data in
logical address $\addr$ is updated to $\data$.
We use standard setting that $\blocksize = \Theta(\log N)$ (so a word can 
store an address).

\medskip
\noindent
{\bf Obliviousness.}
Intuitively, a RAM program $M$ obliviously simulates a RAM program $f$ if: (1) it has the same input/output behavior as $f$; (2) There exists a simulator $\Sim(\abs{x})$ that produces access pattern that is statistically close to the access pattern of $M(x)$, i.e., it can simulate all memory addresses accessed by $M$ during the execution on $x$, without knowing $x$. In case the access pattern and the functionality are randomized, we have to consider the joint distribution of the simulator and the output of the RAM program or the functionality. 


For a  RAM machine $M$ and input $x$, let ${\sf AccPtrn}(M(x))$ denote the distribution of memory addresses a machine $M$ produces on an input $x$.
\begin{definition}
A RAM algorithm $M$ obliviously implements the functionality $f$ with $\e$-obliviousness if the following hold:
\begin{eqnarray*}
\left\{\Sim(1^\lambda),f(x)\right\}_{x \in \bit^\lambda} \overset{\epsilon(\lambda)}{\equiv} \left\{ {\sf AccPtrn}(M(x)),M(x)\right\}_{x \in \bit^\lambda}
\end{eqnarray*}
If $\epsilon(\cdot)=0$, we say $M$ is perfectly oblivious. 
\end{definition}

The two main functionalities that we focus on in this paper are the following:

\paragraph{Oblivious sort:}
This is a deterministic functionality in which the input is an array $A[1,\ldots,n]$ of memory blocks (i.e., each $A[i] \in \bit^\blocksize$, representing a key). The goal is to output an array $A'[1,\ldots,n]$ which is some permutation $\pi:[n] \rightarrow [n]$ of the array $A$, i.e., $A'[i] = A[\pi(i)]$, such that $A'[1]\leq \ldots \leq A'[n]$. 

\paragraph{Oblivious permutation:} 
This is a randomized functionality in which the input is an array $A[1,\ldots,n]$ of memory blocks. The functionality chooses a random permutation $\pi:[n] \rightarrow [n]$ and outputs an array $A'[1,\ldots,n]$ such that $A'[i] = A[\pi(i)]$ for every $i$. 

\section{Our Construction} 
\label{sec:construction}
\label{sec:random-bin-assignment}

We first present the oblivious random bin assignment algorithm (Section~\ref{sec:obin})  and then use it to implement our bucket oblivious random permutation (Section~\ref{sec:ORP}) and bucket oblivious sort (Section~\ref{sec:osort}).

\newcommand{\val}{{\sf value}}
\newcommand{\pref}{{\sf pref}}

\begin{figure*}[h!]
\centering

\begin{algorithm}[Oblivious Random Bin Assignment]
\begin{algorithmic}
\State
\State \textbf{Input}: an array $\X$ of size $n$
\State Choose a bucket size $Z$ and let $B$ be the smallest power of two that is $\geq 2n/Z$. 
\State Define $(\log B+1)$ arrays, each containing $B$ buckets of size $Z$. Denote the $j$-th bucket of the $i$-th array $A_j^{(i)}$.
\State For each element in $\X$, assign a uniformly random key in $[0,B-1]$.
\State Evenly divide $\X$ into $B$ groups. Put the $j$-th group into $A_j^{(0)}$ and pad with dummy elements to have size $Z$.

\For {$i=0,\ldots,\log B-1$}
    \For {$j=0,\ldots,B/2-1$}
        \State $(A^{(i+1)}_{2j}, A^{(i+1)}_{2j+1}) \leftarrow \textsc{MergeSplit}(A^{(i)}_{j'+j}, A^{(i)}_{j'+j+2^i}, i)$ where $j'=\lfloor j / {2^i} \rfloor \cdot 2^{i}$ 
        \State \Comment{Input: $j$-th pair of buckets with distance $2^i$ in $A^{(i)}$; Output: $j$-th pair of buckets in $A^{(i+1)}$}
        
    \EndFor
\EndFor    
\State \textbf{Output:} $A^{(\log B)} = A_0^{(\log B)} \| \ldots A_{B-1}^{(\log B)}$.

\medskip
\Function{$(A'_0, A'_1) \leftarrow$ MergeSplit}{$A_0, A_1, i$}
    \State $A'_0$ receives all real elements in $A_0 \cup A_1$ where the $(i+1)$-st MSB of the key is $0$   
    \State $A'_1$ receives all real elements in $A_0 \cup A_1$ where the $(i+1)$-st MSB of the key is $1$
    \State If either $A'_0$ or $A'_1$ receives more than $Z$ real elements, the procedure aborts with {\sf overflow}
    \State Pad $A'_0$ and $A'_1$ to size $Z$ with dummy elements and return $(A'_0, A'_1)$
\EndFunction   
\end{algorithmic}
\label{code:obin}
\end{algorithm}

\end{figure*}

\subsection{Oblivious Random Bin Assignment}
\label{sec:obin}

The input to the oblivious random bin assignment algorithm is an array $\X$ of $n$ elements. 
The goal is to obliviously and uniformly randomly distribute the elements into a set of bins. Each element is assigned to independent random bin, and elements are then routed into the bins obliviously. 

The algorithm first chooses a bucket size $Z$, which can be set to the security parameter $\sec$. 
Then, it constructs $B=\lceil 2n/Z \rceil$ buckets each of size $Z$.
Without loss of generality, assume $B$ is a power of $2$ --- if not, pad it to the next power of 2. Note that the algorithm introduces $n$ dummy elements, and the output is twice the size of the input array. 


Figure~\ref{fig:radix-sort} gives a graphic illustration of the algorithm for 8 input buckets and Algorithm~\ref{code:obin} gives the pseudocode.
Each element in $\X$ is assigned a random key in $[0, B-1]$ which represents a destination bucket.
Next, the algorithm repeatedly calls the \textsc{MergeSplit} subroutine to exchange elements between bucket pairs in $\log B$ levels to distribute elements into their destination buckets. 
The operation $(A'_0,A'_1)\leftarrow \textsc{MergeSplit}(A_0,A_1,i)$ involves four buckets at the time, distributing the elements in the two input buckets $A_0$ and $A_1$ into two output buckets $A'_0$ and $A'_1$.
$A'_0$ receives all the keys with $(i+1)$-th most significant bit (MSB) as 0 and $A'_1$ receives all the keys with $(i+1)$-th MSB as 1.

For now, assume the client can locally store two buckets.
For each \textsc{MergeSplit}, it reads (and decrypts) the two input buckets, swaps elements in the two buckets according to the above rule, and writes to the two output buckets (after re-encryption).
It is then easy to see that Algorithm~\ref{code:obin} is oblivious since the order in which the client reads and writes the buckets is fixed and independent of the input array.

When no bucket overflows, all real elements are correctly put into their assigned bins.
We now show that the probability of overflow is exponentially small in $Z$. 
Intuitively, this is because each bucket contains (in expectation) half dummy elements that serve as a form of ``slack'' to disallow overflow.

\begin{lemma}
\label{lemma:shuffle}
Overflow happens with at most $\epsilon(n, Z) = 2n/Z \cdot \log(2n/Z) \cdot e^{-Z/6}$ probability.
\end{lemma}
\begin{proof}
\label{clm:proof-shuffle}
Consider a bucket $A^{(i)}_b$ at level $i$.
Observe that this bucket can receive real elements from $2^i$ initial buckets, each containing $Z/2$ real elements.
For each such element, we have chosen an independent and uniformly random key;
the element reaches $A^{(i)}_b$ only when the most significant $i$ bits of its key match $b$,
which happens with exactly $2^{-i}$ probability.
A Chernoff bound shows that $A^{(i)}_b$ overflows with less than $e^{-Z/6}$ probability.
Hence, a union bound over all levels and all buckets 
shows that overflow happens with less than $B \cdot \log B \cdot e^{-Z/6} = \epsilon(n,Z)$ probability.
\end{proof}


\subsection{Bucket Oblivious Random Permutation}
\label{sec:ORP}

After performing the oblivious random bin assignment, ORP can be simply achieved as follows:
scan the array and delete dummy elements from each bin (note that within each bin it is guaranteed that the real elements appear before the dummy elements). Then obliviously permute each bin and finally concatenate all bins.  We have:


\begin{lemma}
Bucket ORP oblivious implement the permutation functionality except for $\e(n,Z)$ probability. 
\end{lemma}

\begin{proof}
We first describe the simulator. 
The access pattern of the oblivious bin assignment algorithm is deterministic and the same for every input, where the overflow even is independent of the input itself. Therefore, it is easy to simulate the bin assignment. 
The simulator then pretends to simulate the randomly permuting of each bin. 
Then, 
the simulator chooses random loads $\vec{k}=(k_0, k_1, \ldots, k_{B-1})$, where $k_i$ is the load of the real elements in the $i$th bin. This is done by simply throwing $n$ elements into $B$ bins (``in the head''). If there is some $i$ for which $k_i > Z$ then the simulator aborts. The removal of the dummy elements is equivalent to the revealing of these loads. 

Clearly, $\vec{k}$ are distributed the same as in the real execution. The only difference between the simulated access pattern and the real one is in the case where the algorithm aborts as a result of an overflow before the last level, which occurs with at most $\e(n,Z)$ probability. 

We next show that the output of the algorithm is a random permutation, conditioned on the access pattern. As we previously described, it is actually enough to condition on the vector of random loads $\vec{k}=(k_0, k_1, \ldots, k_{B-1})$. We show that given any such vector, all permutations are equally likely.  

Fix a particular load $\vec{k}=(k_0, k_1, \ldots, k_{B-1})$. The algorithm works by first assigning the real elements into the bins, and then permuting within each bin. For every input, there are exactly ${n \choose k_0,\ldots,k_{B-1}}$ ways to distribute the real elements into the bins while achieving the vector of loads $\vec{k}$. Then, each bin is individually permuted, i.e., within each bin $i$, we have $k_i$ different possible ordering. Overall, 
the total number of possible outputs with that load is then
\[{n \choose k_0,\ldots,k_{B-1}} \cdot k_0! \cdot \ldots \cdot k_{B-1}! = n!\]
That is, even conditioned on some specific loads $\vec{k}=(k_0, k_1, \ldots, k_{B-1})$, all permutations are still equally likely.
Therefore, $\forall \pi$, $\Pr\left[\Pi = \pi \mid \vec{K}=\vec{k} \right] = \frac 1 {n!}$, and
\[ \Pr\left[\Pi = \pi\right] = \sum_{\vec{k}} \Pr\left[\Pi = \pi \mid \vec{K}=\vec{k} \right] \cdot \Pr\left[\vec{K}=\vec{k}\right] = \frac {1}{n!} \]

Our algorithm fails to implement the ORP only when some bin overflows during the oblivious random bin assignment, which happens with $\epsilon(n,Z)$ probability by Lemma~\ref{lemma:shuffle}.
\end{proof}

\subsection{Bucket Oblivious Sort}
\label{sec:osort}
Once we have ORP, it is easy to achieve oblivious sort: just invoke any non-oblivious comparison-based sort after ORP.

Since the functionality is deterministic, it is enough to consider separately correctness and simulation. Correctness follows from directly from the correctness of the ORP and the non-oblivious sort. As for obliviousness, given any input array, one can easily simulate the algorithm by first randomly permuting the array and then running the comparison-based non-oblivious sort. 
The access patterns of a comparison-based sort depend only on the relative ranking of the input elements, which is independent of the input array once the array has been randomly permuted. 

\subsection{Efficiency}
\label{sec:efficiency}

We analyze the efficiency of our algorithms and compare them to classic non-oblivious oblivious sorting algorithms in Table~\ref{tab:compare}.
We measure runtime using the number of memory accesses the clients needs to perform on the server.


For our algorithms, assuming the client can store $2Z$ elements locally, each $2n$-sized array is read and written once and there are $\log(2n/Z)<\log n$ of them.
So oblivious bin assignment and bucket ORP run in (less than) $4n\log n$ time.
Note that the last step of ORP, i.e., permuting each output bucket, can be incorporated with the last level of oblivious bin assignment.
Bucket oblivious sort additionally invokes a non-oblivious sort, and thus runs in $6n\log n$ time. 
This is within $3\times$ of merge sort and beats bitonic sort when $n$ is moderately large;
for example, $5\times$ faster than bitonic for $n=2^{30}$.
For an overflow probability of $2^{-80}$ and most reasonable values of $n$, $Z = 512$ suffices.


\section{Extensions}
\label{sec:extensions}

\subsection{Extension to Constant Client Storage}
\label{sec:O1client}
We now discuss how to extend our algorithms to the case where the client can only store $O(1)$ elements locally.

Each \textsc{MergeSplit} can be realized with a single invocation of bitonic sort.
Concretely, we first scan the two input buckets to count how many real elements should go to buckets $A'_0$ vs. $A'_1$, then tag the correct number of dummy elements going to either buckets, and finally perform a bitonic sort.

Next, we need to permute each output bucket obliviously with $O(1)$ local storage. 
This can be done as follows. 
First, assign each element in a bucket a uniformly random label of $\Theta(\log n)$ bits. 
Then, obliviously sort the elements by their random labels using bitonic sort. 
Since the labels are ``short'' (i.e., logarithmic in size), we may have collisions with $n^{-c}$ probability for some constant $c$, in which case we simply retry. 
In expectation, it succeeds in $1+o(1)$ trials. 


Since we invoke $B/2$ instances of bitonic sort on $2Z$ elements at each level,
the runtime is roughly $\log B \cdot B/2 \cdot 2Z \log^2 (2Z)) \approx 2 n\log n \log^2 Z$. 

\subsection{Better Asymptotic Performance}
Our algorithms can also be extended to have better asymptotic performance.
For this instantiation, we use a primitive called oblivious tight compaction.
Oblivious tight compaction receives $n$ elements each marked as either 0 or 1, and outputs a permutation of the $n$ elements such that all elements marked 0 appear before the elements that are marked 1. 
It should not be hard to see that oblivious tight compaction can be used to achieve \textsc{MergeSplit}.
Using the $O(1)$-client-storage and $O(n)$-time oblivious tight compaction construction from~\cite{asharov2018optorama}, bucket oblivious sort achieves $O(n\log n + n\log^2Z)$ runtime and $O(1)$ client storage.
Setting $Z=\omega(1)\log n$, bucket oblivious sort achieves $O(n\log n)$ runtime, $O(1)$ client storage, and a negligible in $n$ error probability.

\subsection{Locality}
\label{sec:locality}
Algorithmic performance when the data is stored on disk has been studied in the external disk model (e.g.,~\cite{RuemmlerW94,ArgeFGV97,Vitter01,Vitter06}) and references within). Recently, Asharov et al.~\cite{AsharovCNPRS19} extended this study to oblivious algorithms. In this setting, an algorithm 
is said to have $(p, \ell)$ locality if it has access 
to $p$ disks and 
accesses in total $\ell$ discontiguous memory regions in all disks combined. As an example, it is not hard to see that merge sort is a non-oblivious sorting algorithm that sorts an array of size $n$ in $O(n \log n)$ and $(3,\log n)$-locality, whereas quick sort  is not local for any reasonable $p$. 
This locality metric is motivated by the fact that real-world storage
media such as disks support sequential accesses
much faster than random seeks. Thus an algorithm that 
makes mostly sequential accesses would execute much faster in practice than one that  
makes mostly random accesses --- even if the two have the same runtime in a standard
word-RAM model. 

Guided by this new metric, Asharov et al.~\cite{AsharovCNPRS19} consider how to design oblivious algorithms and ORAM schemes that achieve good locality. 
Since sorting is one of the most important
building blocks in the 
design of oblivious algorithms, 
inevitably Asharov et al.~\cite{AsharovCNPRS19}
show a locality-friendly sorting algorithm.
Concretely, they show that there is a specific way to implement
the bitonic sort meta-algorithm,
such that the entire algorithm requires accessing 
$O(\log^2 n)$ distinct memory regions (i.e., as many as the depth of the sorting network) 
require only 2 disks to be available --- in other words,
the algorithm achieves $(2, O(\log^2 n))$-locality.

We observe that our algorithm, when implemented properly, is a locality-friendly oblivious sorting algorithm. 
Our algorithm 
outperforms Asharov et al.~\cite{AsharovCNPRS19}'s  scheme 
by an almost logarithmic 
factor improvement in locality. 
To achieve this, the crux is to implement all $n/Z$ instances of 
$\textsc{MergeSplit}$ in the same layer of the butterfly network 
while accessing a small number of discontiguous regions. Specifically, the $\textsc{MergeSplit}$ operation works on 4 buckets at a time, while reading two buckets from the input layer, and writing to two consecutive buckets in the output layer. Moreover, the different invocations of $\textsc{MergeSplit}$ on the same layer deal with consecutive buckets. By carefully distributing the buckets among the different disks, and by using bitonic sort while implementing the $\textsc{MergeSplit}$ operation, we conclude:

\begin{corollary}
There exists a statistically oblivious sort algorithm which, except with $\approx e^{-Z/6}$ probability, completes in $O(n \log n \log^2 Z)$ work and with $(3, O(\log n \log^2 Z)$) locality.
\end{corollary}
%



\paragraph{Acknowledgement.} The authors thank Yutong Dai and Peijing Xu for proofreading the manuscript.

\bibliographystyle{alpha}
\bibliography{refs}

\newcommand{\etalchar}[1]{$^{#1}$}
\begin{thebibliography}{LWN{\etalchar{+}}15}

\bibitem[ACN{\etalchar{+}}19]{AsharovCNPRS19}
Gilad Asharov, T-H~Hubert Chan, Kartik Nayak, Rafael Pass, Ling Ren, and Elaine
  Shi.
\newblock Locality-preserving oblivious {RAM}.
\newblock In {\em Annual International Conference on the Theory and
  Applications of Cryptographic Techniques}, pages 214--243. Springer, 2019.

\bibitem[AFGV97]{ArgeFGV97}
Lars Arge, Paolo Ferragina, Roberto Grossi, and Jeffrey~Scott Vitter.
\newblock On sorting strings in external memory (extended abstract).
\newblock In {\em {ACM} Symposium on the Theory of Computing (STOC '97)}, pages
  540--548, 1997.

\bibitem[AKL{\etalchar{+}}18]{asharov2018optorama}
Gilad Asharov, Ilan Komargodski, Wei-Kai Lin, Kartik Nayak, Enoch Peserico, and
  Elaine Shi.
\newblock {OptORAMa:} optimal oblivious {RAM}.
\newblock {\em Cryptology ePrint Archive}, 2018.

\bibitem[AKS83]{aks}
Mikl{\'o}s Ajtai, J{\'a}nos Koml{\'o}s, and Endre Szemer{\'e}di.
\newblock An {$0(n \log n)$} sorting network.
\newblock In {\em Proceedings of the fifteenth annual ACM symposium on Theory
  of computing}, pages 1--9. ACM, 1983.

\bibitem[Bat68]{bitonic}
Kenneth~E Batcher.
\newblock Sorting networks and their applications.
\newblock In {\em Proceedings of the April 30--May 2, 1968, spring joint
  computer conference}, pages 307--314. ACM, 1968.

\bibitem[CV14]{thorp01}
Artur Czumaj and Berthold V{\"{o}}cking.
\newblock Thorp shuffling, butterflies, and non-markovian couplings.
\newblock In {\em {ICALP} {(1)}}, volume 8572 of {\em Lecture Notes in Computer
  Science}, pages 344--355. Springer, 2014.

\bibitem[Czu15]{randpermnet}
Artur Czumaj.
\newblock Random permutations using switching networks.
\newblock In {\em {STOC}}, pages 703--712. {ACM}, 2015.

\bibitem[FNR{\etalchar{+}}15]{fletcher2015bucket}
Christopher~W Fletcher, Muhammad Naveed, Ling Ren, Elaine Shi, and Emil
  Stefanov.
\newblock Bucket {ORAM}: Single online roundtrip, constant bandwidth oblivious
  {RAM}.
\newblock {\em Cryptology ePrint Archive}, 2015.

\bibitem[GM11]{goodrich2011privacy}
Michael~T Goodrich and Michael Mitzenmacher.
\newblock Privacy-preserving access of outsourced data via oblivious {RAM}
  simulation.
\newblock In {\em International Colloquium on Automata, Languages, and
  Programming}, pages 576--587. Springer, 2011.

\bibitem[GO96]{goldreich1996software}
Oded Goldreich and Rafail Ostrovsky.
\newblock Software protection and simulation on oblivious rams.
\newblock {\em Journal of the ACM}, 43(3):431--473, 1996.

\bibitem[Goo10]{RandShellsort}
Michael~T Goodrich.
\newblock Randomized {Shellsort}: A simple oblivious sorting algorithm.
\newblock In {\em Proceedings of the twenty-first annual ACM-SIAM symposium on
  Discrete Algorithms}, pages 1262--1277. SIAM, 2010.

\bibitem[Goo14]{zigzag}
Michael~T Goodrich.
\newblock Zig-zag sort: A simple deterministic data-oblivious sorting algorithm
  running in {$O(n \log n$} time.
\newblock In {\em Proceedings of the forty-sixth annual ACM symposium on Theory
  of computing}, pages 684--693. ACM, 2014.

\bibitem[LWN{\etalchar{+}}15]{oblivm}
Chang Liu, Xiao~Shaun Wang, Kartik Nayak, Yan Huang, and Elaine Shi.
\newblock {ObliVM}: {A} programming framework for secure computation.
\newblock In {\em Symposium on Security and Privacy}. IEEE, 2015.

\bibitem[NWI{\etalchar{+}}15]{graphsc}
Kartik Nayak, Xiao~Shaun Wang, Stratis Ioannidis, Udi Weinsberg, Nina Taft, and
  Elaine Shi.
\newblock {GraphSC}: Parallel secure computation made easy.
\newblock In {\em Symposium on Security and Privacy}. IEEE, 2015.

\bibitem[OGTU14]{ohrimenko2014melbourne}
Olga Ohrimenko, Michael~T Goodrich, Roberto Tamassia, and Eli Upfal.
\newblock The melbourne shuffle: Improving oblivious storage in the cloud.
\newblock In {\em International Colloquium on Automata, Languages, and
  Programming}, pages 556--567. Springer, 2014.

\bibitem[RS20]{Ramachandran}
Vijaya Ramachandran and Elaine Shi.
\newblock Data oblivious algorithms for multicores.
\newblock {\em CoRR}, abs/2008.00332, 2020.

\bibitem[RW94]{RuemmlerW94}
Chris Ruemmler and John Wilkes.
\newblock An introduction to disk drive modeling.
\newblock {\em {IEEE} Computer}, 27(3):17--28, 1994.

\bibitem[SS13]{oblivistore}
Emil Stefanov and Elaine Shi.
\newblock Oblivistore: High performance oblivious cloud storage.
\newblock In {\em Symposium on Security and Privacy}. IEEE, 2013.

\bibitem[Vit01]{Vitter01}
Jeffrey~Scott Vitter.
\newblock External memory algorithms and data structures.
\newblock {\em {ACM} Comput. Surv.}, 33(2):209--271, 2001.

\bibitem[Vit06]{Vitter06}
Jeffrey~Scott Vitter.
\newblock Algorithms and data structures for external memory.
\newblock {\em Foundations and Trends in Theoretical Computer Science},
  2(4):305--474, 2006.

\end{thebibliography}

\appendix

\end{document}